\pdfoutput=1
\RequirePackage{ifpdf}
\ifpdf 
\documentclass[pdftex]{sigma}
\else
\documentclass{sigma}
\fi

\numberwithin{equation}{section}

\newtheorem{Theorem}{Theorem}[section]

\newtheorem{Proposition}[Theorem]{Proposition}
\newtheorem{Conjecture}[Theorem]{Conjecture}
 { \theoremstyle{definition}

\newtheorem{Remark}[Theorem]{Remark} }

\usepackage{pb-diagram}

\begin{document}
\allowdisplaybreaks

\newcommand{\arXivNumber}{1811.03285}

\renewcommand{\PaperNumber}{074}

\FirstPageHeading

\ShortArticleName{Combinatorial Expressions for the Tau Functions of $q$-Painlev\'e V and III Equations}

\ArticleName{Combinatorial Expressions for the Tau Functions\\
of $\boldsymbol{q}$-Painlev\'e V and III Equations}

\Author{Yuya MATSUHIRA and Hajime NAGOYA}

\AuthorNameForHeading{Y.~Matsuhira and H.~Nagoya}

\Address{School of Mathematics and Physics, Kanazawa University, Kanazawa, Ishikawa 920-1192, Japan}
\Email{\href{mailto:y.matsu0727@gmail.com}{y.matsu0727@gmail.com}, \href{mailto:nagoya@se.kanazawa-u.ac.jp}{nagoya@se.kanazawa-u.ac.jp}}

\ArticleDates{Received November 24, 2018, in final form September 13, 2019; Published online September 23, 2019}

\Abstract{We derive series representations for the tau functions of the $q$-Painlev\'e~V,~$\mathrm{III_1}$, $\mathrm{III_2}$, and $\mathrm{III_3}$ equations, as degenerations of the tau functions of the $q$-Painlev\'e VI equation in~[Jimbo M., Nagoya H., Sakai H., \textit{J.~Integrable Syst.} \textbf{2} (2017), xyx009, 27~pages]. Our tau functions are expressed in terms of $q$-Nekrasov functions. Thus, our series representations for the tau functions have explicit combinatorial structures. We show that general solutions to the $q$-Painlev\'e V, $\mathrm{III_1}$, $\mathrm{III_2}$, and $\mathrm{III_3}$ equations are written by our tau functions. We also prove that our tau functions for the $q$-Painlev\'e $\mathrm{III_1}$, $\mathrm{III_2}$, and~$\mathrm{III_3}$ equations satisfy the three-term bilinear equations for them.}


\Classification{39A13; 33E17; 05A30}

\section{Introduction}\label{section1}
The $q$-Painlev\'e equations \cite{KNY, RGH} are $q$-difference analogs of the Painlev\'e equations, which were introduced as new special functions beyond elliptic functions and the hypergeometric functions more than one hundred years ago \cite{Gambier, P,P+}, and are now considered as important special functions with many applications both in mathematics and physics.

Similarly, as for other integrable systems, tau functions play a crucial role in the studies of the Painelv\'e equations. The recent discovery by \cite{GIL} states that the tau function of the sixth Painlev\'e equation is a Fourier transform of Virasoro conformal blocks with $c=1$, which admit explicit combinatorial formulas by AGT correspondence~\cite{AGT}. Series representations of the tau functions of other types are also studied in \cite{BLMST,GIL1,N,N1} for differential cases, \cite{BS,BS1,JNS} for $q$-difference cases.

In \cite{JNS}, a general solution ($y,z$) to the $q$-Painlev\'e VI equation \cite{JS} was expressed by the tau functions having $q$-Nekrasov type expressions,
and it was conjectured that the tau functions satisfy the bilinear equations
for the $q$-Painlev\'e VI equation.
In this paper, we give explicit expressions for general solutions
to the $q$-Painlev\'e V, $\mathrm{III_1}$, $\mathrm{III_2}$, and $\mathrm{III_3}$
 equations using degenerations of the tau functions
of the $q$-Painlev\'e VI equation.
We also give conjectures on the bilinear equations satisfied by the tau functions
of the $q$-Painlev\'e V equation and prove that the tau functions of the $q$-Painlev\'e
$\mathrm{III_1}$, $\mathrm{III_2}$, and $\mathrm{III_3}$ equations
satisfy the bilinear equations.

Our $q$-difference equations are as follows.

\noindent
(i) the $q$-Painlev\'e VI equation:
\begin{gather*}
 \frac{y\overline{y}}{a_3a_4}=
\frac{(\overline{z}-b_1t)(\overline{z}-b_2t)}{(\overline{z}-b_3)(\overline{z}-b_4)} ,
\qquad
\frac{z\overline{z}}{b_3b_4}=
\frac{(y-a_1t)(y-a_2t)}{(y-a_3)(y-a_4)} .
\end{gather*}
(ii) the $q$-Painlev\'e V equation:
\begin{gather*}
\frac{y \overline y}{a_3 a_4}=-\frac{(\overline z-b_1 t)(\overline z -b_2 t)}{\overline z-b_3},\qquad
\frac{z \overline z}{b_3}=-\frac{(y-a_1t)(y-a_2 t)}{a_4(y-a_3)}.
\end{gather*}
(iii) the $q$-Painlev\'e $\mathrm{III}_1$ equation:
\begin{gather*}
\frac{y \overline y}{a_3 a_4}=-\frac{\overline z(\overline z -b_2 t)}{\overline z-b_3},\qquad
\frac{z \overline z}{b_3}=-\frac{y(y-a_1 t)}{a_4(y-a_3)}.
\end{gather*}
(iv) the $q$-Painlev\'e $\mathrm{III}_2$ equation:
\begin{gather*}
\frac{y \overline y}{a_3 a_4}=-\frac{\overline z^2}{\overline z-b_3} ,\qquad
\frac{z \overline z}{b_3}=-\frac{y(y-a_2 t)}{a_4(y-a_3)}.
\end{gather*}
(v-1) the $q$-Painlev\'e $\mathrm{III}_3$ equation of surface type $A_7^{(1)\prime}$:
\begin{gather*}
\frac{y \overline y}{a_3 }=\overline z^2 ,\qquad
z \overline z=-\frac{y(y-a_2 t)}{y-a_3}.
\end{gather*}
(v-2) the $q$-Painlev\'e $\mathrm{III}_3$ equation of surface type $A_7^{(1)}$:
\begin{gather*}
\frac{y \overline y}{a_3}=-\frac{\overline z^2}{\overline z-b_3} ,\qquad
z \overline z=\frac{y(y-a_2 t)}{a_2}.
\end{gather*}
Here, $y$, $z$ are functions of $t$, $\overline y=y(qt)$, $\overline z=z(qt)$, and $a_i$, $b_i$ ($i=1,2,3,4$) are parameters.

From the point of view of Sakai's classification for the discrete Painlev\'e equations~\cite{S}, the $q$-Painlev\'e VI, V, $\mathrm{III}_1$, $\mathrm{III}_2$ and $\mathrm{III}_3$ equations are derived from the symmetries/surfaces of type $D_5^{(1)}/A_3^{(1)}$, $A_4^{(1)}/A_4^{(1)}$, $E_2^{(1)}/A_5^{(1)}$, $E_2^{(1)}/A_6^{(1)}$ and $A_1^{(1)}/A_7^{ (1)}$, respectively.

The degeneration scheme of Painlev\'e equations is as follows
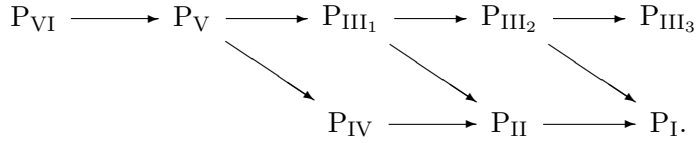
\begin{gather*}
\begin{diagram}
\node{\mathrm{P_{VI}}}\arrow{e}
\node{\mathrm{P_{V}}}\arrow{e}\arrow{se}
\node{\mathrm{P_{III_1}}}\arrow{e}\arrow{se}
\node{\mathrm{P_{III_2}}}\arrow{e}\arrow{se}
\node{\mathrm{P_{III_3}}}
\\
\node[3]{\mathrm{P_{IV}}}\arrow{e}
\node{\mathrm{P_{II}}}\arrow{e}
\node{\mathrm{P_I}.}
\end{diagram}
\end{gather*}
The degeneration pattern of the $q$-Painlev\'e equations we use is similar to the one in~\cite{Mu} but not exactly the same. Rather, our limiting procedure is a $q$-version for the one used in~\cite{GIL1} in order to derive combinatorial expressions of tau functions of~$\mathrm{P_V}$, $\mathrm{P_{III_1}}$, $\mathrm{P_{III_2}}$, and $\mathrm{P_{III_3}}$
from the Nekrasov type expression of the tau function of~$\mathrm{P_{VI}}$~\cite{GIL}.

For the case of the $q$-Painlev\'e $\mathrm{III}_3$ equation of surface type $A_7^{(1)\prime}$, a series representation for the tau function was proposed in~\cite{BS}, which are
expressed by $q$-Virasoro Whittaker conformal blocks which equal Nekrasov partition functions for pure ${\rm SU}(2)$ 5d theory \cite{AY,Y}. A~Fredholm determinant representation of the tau function by topological strings/spectral theory duality is proposed in~\cite{BGT}. For the $q$-Painlev\'e $\mathrm{III}_3$ equation of surface type~$A_7^{(1)}$, a series representation for the tau function was proposed in~\cite{BGM}. Our tau functions for the $q$-Painlev\'e $\mathrm{III}_3$ equations obtained by the degeneration are equivalent to them.

Our plan is as follows. In Section~\ref{section2}, we recall the result on $q$-Painlev\'e VI equation in \cite{JNS}. In Sections~\ref{section3}--\ref{section6}, we compute limits of tau functions and derive combinatorial expressions of general solutions and bilinear equations for $q$-Painlev\'e~V, $\mathrm{III}_1$, $\mathrm{III}_2$ and $\mathrm{III}_3$ equations.

{\it Notations.} Throughout the paper we fix $q\in\mathbb{C}^\times$ such that $|q|<1$. We set
\begin{gather*}
 [u]=\big(1-q^u\big)/(1-q),\qquad (a;q)_N=\prod_{j=0}^{N-1}\big(1-aq^j\big),\\
(a_1,\dots,a_k;q)_{\infty}=\prod_{j=1}^k(a_j;q)_\infty,\qquad (a;q,q)_\infty=\prod_{j,k=0}^\infty\big(1-aq^{j+k}\big).
\end{gather*}
We use the $q$-Gamma function and $q$-Barnes function defined by
\begin{gather*}
 \Gamma_q(u)=\frac{(q;q)_\infty}{(q^u;q)_\infty}(1-q)^{1-u}, \qquad
G_q(u)=\frac{(q^u;q,q)_\infty}{(q;q,q)_\infty}(q;q)_\infty^{u-1}(1-q)^{-(u-1)(u-2)/2},
\end{gather*}
which satisfy $\Gamma_q(1)=G_q(1)=1$ and
\begin{gather}
 \Gamma_q(u+1)=[u]\Gamma_q(u), \qquad G_q(u+1)=\Gamma_q(u)G_q(u). \label{eq_Gamma_Barnes}
\end{gather}
A partition is a finite sequence of positive integers $\lambda=(\lambda_1,\dots,\lambda_l)$ such that $\lambda_1\ge\dots\ge\lambda_{l}>0$. Denote the length of the partition by $\ell(\lambda)=l$. The conjugate partition $\lambda'=(\lambda'_1,\dots,\lambda'_{l'})$ is defined by $\lambda'_j=\sharp\{i\,|\, \lambda_i\ge j\}$, $l'=\lambda_1$. We regard a partition as a Young diagram. Namely, we regard a partition $\lambda$ also as the subset $\{(i,j)\in\mathbb{Z}^2\,|\, 1\le j\le \lambda_i,\, i\ge 1\}$ of $\mathbb{Z}^2$, and denote its cardinality by $|\lambda|$. We denote the set of all partitions by~$\mathbb{Y}$. For $\square=(i,j)\in\mathbb{Z}_{>0}^2$ we set $a_\lambda(\square)=\lambda_i-j$
(the arm length of $\square$) and $\ell_\lambda(\square)=\lambda'_j-i$ (the leg length of $\square$). In the last formulas we set $\lambda_i=0$ if $i>\ell(\lambda)$
(resp.~$\lambda'_j=0$ if $j>\ell(\lambda')$). For a pair of partitions ($\lambda, \mu$) and $u\in\mathbb{C}$ we set
\begin{gather*}
N_{\lambda,\mu}(u)=\prod_{\square\in \lambda}\big( 1-q^{-\ell_{\lambda}(\square)-a_\mu(\square)-1}u\big)
\prod_{\square\in \mu}\big(1-q^{\ell_{\mu}(\square)+a_\lambda(\square)+1}u\big),
\end{gather*}
which we call a Nekrasov factor.

\section[Results on $q$-$\mathrm{P_{VI}}$ from \cite{JNS}]{Results on $\boldsymbol{q}$-$\boldsymbol{\mathrm{P_{VI}}}$ from \cite{JNS}}\label{section2}

In this section, we recall the results of \cite{JNS} on the $q$-Painlev\'e VI equation. Define the tau function by
\begin{gather*}
 \tau^{\mathrm{VI}}\left[
\begin{matrix}
 \theta_1 & \theta_t \\
 \theta_\infty & \theta_0
\end{matrix}\Bigl|s,\sigma,t\right]
=\sum_{n\in\mathbb{Z}}
s^n t^{(\sigma+n)^2-\theta_t^2-\theta_0^2}
C\left[\begin{matrix}
\theta_1 & \theta_t \\
 \theta_\infty & \theta_0\\
\end{matrix}\Bigl|\sigma+n\right]
Z\left[\begin{matrix}
\theta_1 & \theta_t\\
\theta_\infty &\theta_0 \\
\end{matrix}\Bigl|\sigma+n,t\right] ,
\end{gather*}
with the definition
\begin{gather*}
 C\left[\begin{matrix}
\theta_1 & \theta_t \\
\theta_\infty & \theta_0
\end{matrix}\Bigl|
\sigma\right]
=\frac{\prod\limits_{\varepsilon,\varepsilon'=\pm}G_q(1+\varepsilon\theta_\infty-\theta_1+\varepsilon'\sigma)
G_q(1+\varepsilon\sigma-\theta_t+\varepsilon'\theta_0)}{G_q(1+2\sigma)G_q(1-2\sigma)}
 ,
\\
 Z\left[\begin{matrix}
 \theta_1 & \theta_t\\
 \theta_\infty & \theta_0
\end{matrix}\Bigl|
\sigma,t\right]
=\sum_{\boldsymbol{\lambda}=(\lambda_+,\lambda_-)\in\mathbb{Y}^2}
t^{|\boldsymbol{\lambda}|}\cdot
\frac{\prod\limits_{\varepsilon,\varepsilon'=\pm}
N_{\varnothing,\lambda_{\varepsilon'}}\big(q^{\varepsilon\theta_\infty-\theta_1-\varepsilon'\sigma}\big)
N_{\lambda_{\varepsilon},\varnothing}\big(q^{\varepsilon\sigma-\theta_t-\varepsilon'\theta_0}\big)}
{\prod\limits_{\varepsilon,\varepsilon'=\pm}N_{\lambda_\varepsilon,\lambda_{\varepsilon'}}\big(q^{(\varepsilon-\varepsilon')\sigma}\big)} .
\end{gather*}

Put
\begin{alignat*}{3}
& \tau_1^{\mathrm{VI}} =\tau^{\mathrm{VI}}\left[
\begin{matrix}
 \theta_1 & \theta_t \\
 \theta_\infty+\tfrac{1}{2} & \theta_0\\
\end{matrix}\Bigl|s,\sigma,t\right] ,
\qquad &&
\tau_2^{\mathrm{VI}} =\tau^{\mathrm{VI}}\left[
\begin{matrix}
 \theta_1 & \theta_t \\
 \theta_\infty-\tfrac{1}{2} & \theta_0\\
\end{matrix}\Bigl|s,\sigma,t\right] ,&\\
& \tau_3^{\mathrm{VI}} =\tau^{\mathrm{VI}}\left[
\begin{matrix}
 \theta_1 & \theta_t \\
 \theta_\infty & \theta_0+\tfrac{1}{2}\\
\end{matrix}\Bigl|s,\sigma+\tfrac{1}{2} ,t\right] ,
\qquad &&
\tau_4^{\mathrm{VI}} =\tau^{\mathrm{VI}}\left[
\begin{matrix}
 \theta_1 & \theta_t \\
 \theta_\infty & \theta_0-\tfrac{1}{2}\\
\end{matrix}\Bigl|s,\sigma-\tfrac{1}{2} ,t\right] ,&
\\
& \tau_5^{\mathrm{VI}}=\tau^{\mathrm{VI}}\left[
\begin{matrix}
 \theta_1-\tfrac{1}{2} & \theta_t \\
 \theta_\infty & \theta_0\\
\end{matrix}\Bigl|s,\sigma,t\right] ,
\qquad & &
\tau_6^{\mathrm{VI}} =\tau^{\mathrm{VI}}\left[
\begin{matrix}
 \theta_1+\tfrac{1}{2} & \theta_t \\
 \theta_\infty & \theta_0\\
\end{matrix}\Bigl|s,\sigma,t\right] ,&
\\
& \tau_7^{\mathrm{VI}} =\tau^{\mathrm{VI}}\left[
\begin{matrix}
 \theta_1 & \theta_t-\tfrac{1}{2} \\
 \theta_\infty & \theta_0\\
\end{matrix}\Bigl|s,\sigma+\tfrac{1}{2} ,t\right] ,
\qquad
&&
\tau_8^{\mathrm{VI}} =\tau^{\mathrm{VI}}\left[
\begin{matrix}
 \theta_1 & \theta_t+\tfrac{1}{2} \\
 \theta_\infty & \theta_0\\
\end{matrix}\Bigl|s,\sigma-\tfrac{1}{2} ,t\right] .&
\end{alignat*}

Here and after we write $\overline{f}(t)=f(qt)$, $\underline{f}(t)=f(t/q)$.
\begin{Theorem}[\cite{JNS}]\label{thm_qPVI}
The functions $y$ and $z$ defined by
\begin{gather}\label{eq_yz_thm}
y= q^{-2\theta_1-1}t\cdot \frac{\tau_3^{\mathrm{VI}}\tau_4^{\mathrm{VI}}}{\tau_1^{\mathrm{VI}}\tau_2^{\mathrm{VI}}} ,
\qquad
z=\frac{\underline{\tau_1^{\mathrm{VI}}}\tau_2^{\mathrm{VI}}-\tau_1^{\mathrm{VI}}\underline{\tau_2^{\mathrm{VI}}}}
{q^{1/2+\theta_\infty} \underline{\tau_1^{\mathrm{VI}}}\tau_2^{\mathrm{VI}}-
q^{1/2-\theta_\infty }\tau_1^{\mathrm{VI}}\underline{\tau_2^{\mathrm{VI}}}}
\end{gather}
are solutions to the $q$-Painlev\'e VI equation
\begin{gather}
 \frac{y\overline{y}}{a_3a_4}=
\frac{(\overline{z}-b_1t)(\overline{z}-b_2t)}{(\overline{z}-b_3)(\overline{z}-b_4)} ,
\qquad
\frac{z\overline{z}}{b_3b_4}=
\frac{(y-a_1t)(y-a_2t)}{(y-a_3)(y-a_4)} ,\label{eq_qPVI}
\end{gather}
with the parameters
\begin{gather*}
 a_1=q^{-2\theta_1-1},\qquad a_2=q^{-2\theta_t-2\theta_1-1} ,
\qquad a_3=q^{-1},\qquad a_4=q^{-2\theta_1-1},\\
 b_1=q^{-\theta_0-\theta_t-\theta_1},\qquad
b_2=q^{\theta_0-\theta_t-\theta_1},\qquad
\qquad b_3=q^{\theta_\infty-1/2},\qquad
b_4= q^{-\theta_\infty-1/2} .
\end{gather*}
\end{Theorem}
The formula for $y$ above can be regarded as an extension of Mano's asymptotic expansion to all orders for the solution of $q$-$\mathrm{P_{VI}}$~\cite{Ma}. Theorem \ref{thm_qPVI} was obtained by constructing the fundamental solution of the Lax-pair for $q$-$\mathrm{P_{VI}}$ in~\cite{JS}, in terms of $q$-conformal blocks in~\cite{AFS}. The method of construction of the fundamental solution is a $q$-analogue of the CFT approach used in~\cite{ILT}. In the derivation of Theorem \ref{thm_qPVI} convergence of the fundamental solution was assumed and it has not been proved. Recently, analyticity of K-theoretic Nekrasov functions in cases different from our case was discussed in~\cite{FM}. We remark that the convergence of the pure $q$-Nekrasov partition function with $q_1q_2=1$ on $\mathbb{C}$ is proved in~\cite{BS}.

\begin{Conjecture}[\cite{JNS}]\label{conj_qPVI} The tau functions $\tau_i^{\mathrm{VI}}$ $(i=1,\dots, 8)$ satisfy the following bilinear equations
\begin{gather}
 \tau_1^{\mathrm{VI}}\tau_2^{\mathrm{VI}}-q^{-2\theta_1}t \tau_3^{\mathrm{VI}}\tau_4^{\mathrm{VI}}
-\big(1-q^{-2\theta_1}t\big)\tau_5^{\mathrm{VI}} \tau_6^{\mathrm{VI}}=0,
\label{bilin-1}\\
 \tau_1^{\mathrm{VI}}\tau_2^{\mathrm{VI}}-t \tau_3^{\mathrm{VI}}\tau_4^{\mathrm{VI}}
-\big(1-q^{-2\theta_t}t\big) \underline{\tau_5^{\mathrm{VI}}}\overline{\tau_6^{\mathrm{VI}}}=0,
\label{bilin-2}\\
 \tau_1^{\mathrm{VI}}\tau_2^{\mathrm{VI}}-\tau_3^{\mathrm{VI}}\tau_4^{\mathrm{VI}}+\big(1-q^{-2\theta_1}t\big)q^{2\theta_t}\underline{\tau_7^{\mathrm{VI}}}\overline{\tau_8^{\mathrm{VI}}}=0,
\label{bilin-3}\\
 \tau_1^{\mathrm{VI}}\tau_2^{\mathrm{VI}}-q^{2\theta_t}\tau_3^{\mathrm{VI}}\tau_4^{\mathrm{VI}}+
\big(1-q^{-2\theta_t}t\big)q^{2\theta_t}\tau_7^{\mathrm{VI}} \tau_8^{\mathrm{VI}}=0,
\label{bilin-4}\\
 \underline{\tau_5^{\mathrm{VI}}}\tau_6^{\mathrm{VI}}+q^{-\theta_1-\theta_\infty+\theta_t-1/2}t \underline{\tau_7^{\mathrm{VI}}}\tau_8^{\mathrm{VI}}
-\underline{\tau_1^{\mathrm{VI}}}\tau_2^{\mathrm{VI}}=0,
\label{bilin-5}\\
 \underline{\tau_5^{\mathrm{VI}}}\tau_6^{\mathrm{VI}}+q^{-\theta_1+\theta_\infty+\theta_t-1/2}t \underline{\tau_7^{\mathrm{VI}}}\tau_8^{\mathrm{VI}}
-\tau_1^{\mathrm{VI}}\underline{\tau_2^{\mathrm{VI}}}=0,
\label{bilin-6}\\
 \underline{\tau_5^{\mathrm{VI}}}\tau_6^{\mathrm{VI}}+q^{\theta_0+2\theta_t}\underline{\tau_7^{\mathrm{VI}}}\tau_8^{\mathrm{VI}}
-q^{\theta_t} \underline{\tau_3^{\mathrm{VI}}}\tau_4^{\mathrm{VI}}=0,
\label{bilin-7}\\
 \underline{\tau_5^{\mathrm{VI}}}\tau_6^{\mathrm{VI}}+q^{-\theta_0+2\theta_t}\underline{\tau_7^{\mathrm{VI}}}\tau_8^{\mathrm{VI}}
-q^{\theta_t} \tau_3^{\mathrm{VI}}\underline{\tau_4^{\mathrm{VI}}}=0 .
\label{bilin-8}
\end{gather}
Then, the function $y,z$
\begin{gather}\label{yz-final}
y=q^{-2\theta_1-1}t \frac{\tau_3^{\mathrm{VI}}\tau_4^{\mathrm{VI}}}{\tau_1^{\mathrm{VI}}\tau_2^{\mathrm{VI}}},
\qquad
z=-q^{\theta_t-\theta_1-1} t
\frac{\underline{\tau_7^{\mathrm{VI}}}\tau_8^{\mathrm{VI}}}{\underline{\tau_5^{\mathrm{VI}}}\tau_6^{\mathrm{VI}}}
\end{gather}
solves $q$-$\mathrm{P_{VI}}$ \eqref{eq_qPVI}.
\end{Conjecture}
The function $y$ in Conjecture \ref{conj_qPVI} is expressed in the same form in Theorem~\ref{thm_qPVI}, while the function $z$ in Conjecture~\ref{conj_qPVI} is not. By the bilinear equations~\eqref{bilin-5} and~\eqref{bilin-6}, we obtain the expression of $z$ in~\eqref{yz-final} from the expression of $z$ in~\eqref{eq_yz_thm}.

We note that in \cite{JNS} we have a Lax pair with respect to the shift $t\to qt$, namely, a fundamental solution of the linear $q$-difference equations
 \begin{gather}\label{eq_qPVI_Lax}
 Y(qx,t)=A(x,t)Y(x,t),\qquad Y(x,qt)=B(x,t)Y(x,t)
 \end{gather}
for certain 2 by 2 matrices $A(x,t)$ and $B(x,t)$ was constructed in terms of $q$-Nekrasov functions. From \eqref{eq_qPVI_Lax} we obtain the four-term bilinear equation in \cite[Remark~3.5]{JNS}:
\begin{gather}\label{eq_qPVI_4termbilinear}
\underline{\tau_1^{\mathrm{VI}}}\tau_2^{\mathrm{VI}}-\tau_1^{\mathrm{VI}}\underline{\tau_2^{\mathrm{VI}}}
 =\frac{q^{1/2+\theta_\infty}-q^{1/2-\theta_\infty}}
 {q^{-\theta_0}-q^{\theta_0}}q^{-\theta_1-1}t
 \big(\underline{\tau_3^{\mathrm{VI}}}\tau_4^{\mathrm{VI}}-\tau_3^{\mathrm{VI}}\underline{\tau_4^{\mathrm{VI}}}\big) .
\end{gather}

\section[From $q$-$\mathrm{P_{VI}}$ to $q$-$\mathrm{P_V}$]{From $\boldsymbol{q}$-$\boldsymbol{\mathrm{P_{VI}}}$ to $\boldsymbol{q}$-$\boldsymbol{\mathrm{P_V}}$}\label{section3}

In this section, we take a limit of the tau functions of $q$-$\mathrm{P_{VI}}$ to $q$-$\mathrm{P_V}$. Define the tau function by
\begin{gather*}
\tau^{\mathrm{V}}(\theta_*,\theta_t,\theta_0\,|\, s,\sigma,t)=\sum_{n\in\mathbb{Z}}s^n t^{(\sigma+n)^2-\theta_t^2-\theta_0^2}
C_{\mathrm{V}} [\theta_*,\theta_t,\theta_0\,|\,\sigma+n ]
 Z_{\mathrm{V}} [\theta_*,\theta_t,\theta_0\,|\,\sigma+n,t ] ,
\end{gather*}
with
\begin{gather*}
C_{\mathrm{V}}\left[\theta_*,\theta_t,\theta_0\,|\,\sigma\right]
= (q-1)^{-\sigma^2}
\prod_{\varepsilon=\pm} \frac{G_q(1-\theta_*+\varepsilon\sigma)}{G_q(1+2\varepsilon\sigma)}
\prod_{\varepsilon,\varepsilon'=\pm}G_q(1+\varepsilon\sigma-\theta_t+\varepsilon'\theta_0),
\\
Z_{\mathrm{V}}\left[\theta_*,\theta_t,\theta_0\,|\,\sigma,t\right]=
\!\sum_{(\lambda_+,\lambda_-)\in\mathbb{Y}^2}\!\!
t^{|\lambda_+|+|\lambda_-|}
\frac{\prod\limits_{\varepsilon=\pm}\! N_{\varnothing,\lambda_{\varepsilon}}
(q^{-\theta_*-\varepsilon \sigma})
f_{\lambda_{\varepsilon}}(q^{\varepsilon\sigma})
\!\prod\limits_{\varepsilon,\varepsilon'=\pm}\!
N_{\lambda_{\varepsilon},\varnothing}(q^{\varepsilon\sigma-\theta_t-\varepsilon'\theta_0})}
{\prod\limits_{\varepsilon,\varepsilon'=\pm}N_{\lambda_\varepsilon,\lambda_{\varepsilon'}}(q^{(\varepsilon-\varepsilon')\sigma})} ,
\end{gather*}
where
\begin{gather*}
f_{\lambda}(u)=\prod_{\square\in\lambda}\big(-q^{\ell_\lambda(\square)+a_\varnothing(\square)+1}u^{-1}\big).
\end{gather*}
We remark that the factor $f_{\lambda}(u)$ corresponds to the five-dimensional Chern--Simons term. The Chern--Simons term in~\cite{T} reads as
 \begin{gather*}
 \exp\bigg({-}\beta \sum_k \sum_{(i,j)\in Y_k} (a_k+\epsilon (i-j))\bigg),
 \end{gather*}
where $\beta$, $a_k$ are parameters and $Y_1,\dots, Y_N$ are Young tableaux labelling the fixed points. See~\cite{T} for the details. Since
\begin{gather*}
\sum_{\square\in\lambda}\ell_\lambda(\square)+a_\varnothing(\square)+1= \sum_{(i,j)\in\lambda}\lambda_j'-i-j+1 =\sum_{(i,j)\in\lambda}i-j,
\end{gather*}
they coincide when $N=2$. It is possible to remove $f_{\lambda_\varepsilon}(q^{\varepsilon\sigma})$ from $Z_\mathrm{V} [\theta_*,\theta_t,\theta_0\,|\,\sigma,t ]$ by change of variables. Because if we set
\begin{gather*}
Z_\mathrm{V}^{CS=0} [\theta_*,\theta_t,\theta_0\,|\,\sigma,t ]
=\sum_{(\lambda_+,\lambda_-)\in\mathbb{Y}^2}
t^{|\lambda_+|+|\lambda_-|}
\frac{\prod\limits_{\varepsilon=\pm}N_{\varnothing,\lambda_{\varepsilon}}
\big(q^{-\theta_*-\varepsilon \sigma}\big)
\prod\limits_{\varepsilon,\varepsilon'=\pm}
N_{\lambda_{\varepsilon},\varnothing}\big(q^{\varepsilon\sigma-\theta_t-\varepsilon'\theta_0}\big)}
{\prod\limits_{\varepsilon,\varepsilon'=\pm}N_{\lambda_\varepsilon,\lambda_{\varepsilon'}}(q^{(\varepsilon-\varepsilon')\sigma})},
\end{gather*}
then we have
\begin{gather*}
Z_\mathrm{V} [\theta_*,\theta_t,\theta_0\,|\,\sigma,t ]
=Z_\mathrm{V}^{CS=0}\big[{-}\theta_*,-\theta_t,\theta_0\,|\,\sigma,q^{-\theta_*-2\theta_t}t\big]
\end{gather*}
from the relations $N_{\varnothing,\lambda}(u)=f_\lambda\big(u^{-1}\big)N_{\lambda,\varnothing}\big(u^{-1}\big)$, $N_{\lambda,\varnothing}(u)=f_\lambda(u)^{-1}N_{\varnothing,\lambda}\big(u^{-1}\big)$,
and $N_{\lambda,\mu}(u)=N_{\mu',\lambda'}(u)$ \cite[Lemma~A.2]{JNS}.

We define tau functions for $q$-$\mathrm{P_V}$ by
\begin{alignat*}{3}
&\tau_1^{\mathrm{V}}=\tau^{\mathrm{V}}\big(\theta_*-\tfrac{1}{2}, \theta_t,\theta_0\,|\, s,\sigma,t/\sqrt{q}\big),\qquad&&
\tau_2^{\mathrm{V}}=\tau^{\mathrm{V}}\big(\theta_*+\tfrac{1}{2}, \theta_t,\theta_0\,|\, s,\sigma,\sqrt{q}t\big),&\\
&\tau_3^{\mathrm{V}}=\tau^{\mathrm{V}}\big(\theta_*, \theta_t,\theta_0+\tfrac{1}{2}\,|\, s,\sigma+\tfrac{1}{2},t\big),\qquad&&
\tau_4^{\mathrm{V}}=\tau^{\mathrm{V}}\big(\theta_*, \theta_t,\theta_0-\tfrac{1}{2}\,|\, s,\sigma-\tfrac{1}{2},t\big),&\\
&\tau_5^{\mathrm{V}}=\tau^{\mathrm{V}}\big(\theta_*, \theta_t-\tfrac{1}{2},\theta_0\,|\, s,\sigma+\tfrac{1}{2},t\big),\qquad&&
\tau_6^{\mathrm{V}}=\tau^{\mathrm{V}}\big(\theta_*, \theta_t+\tfrac{1}{2},\theta_0\,|\, s,\sigma-\tfrac{1}{2},t\big).&
\end{alignat*}
Let
\begin{gather*}
C_1= C_6=(q-1)^{-\sigma^2}
q^{-(\Lambda+1/2)(\sigma^2-\theta_t^2-\theta_0^2)}
\prod_{\varepsilon=\pm}G_q\big(\tfrac{1}{2}-\Lambda+\varepsilon \sigma\big)^{-1},\\
C_2= C_5=(q-1)^{-\sigma^2}
q^{-(\Lambda-1/2)(\sigma^2-\theta_t^2-\theta_0^2)}
\prod_{\varepsilon=\pm}G_q\big(\tfrac{3}{2}-\Lambda+\varepsilon \sigma\big)^{-1},\\
C_3= (q-1)^{-(\sigma+1/2)^2}
q^{-\Lambda((\sigma+1/2)^2-\theta_t^2-(\theta_0+1/2)^2)}
\prod_{\varepsilon=\pm}G_q\big(1-\Lambda+\varepsilon \big(\sigma+\tfrac{1}{2}\big)\big)^{-1},\\
C_4= (q-1)^{-(\sigma-1/2)^2}
q^{-\Lambda((\sigma-1/2)^2-\theta_t^2-(\theta_0-1/2)^2)}
\prod_{\varepsilon=\pm}G_q\big(1-\Lambda+\varepsilon\big( \sigma-\tfrac{1}{2}\big)\big)^{-1},\\
C_7= (q-1)^{-(\sigma+1/2)^2}
q^{-\Lambda((\sigma+1/2)^2-(\theta_t-1/2)^2-\theta_0^2)}
\prod_{\varepsilon=\pm}G_q\big(1-\Lambda+\varepsilon \big(\sigma+\tfrac{1}{2}\big)\big)^{-1},\\
C_8= (q-1)^{-(\sigma-1/2)^2}
q^{-\Lambda((\sigma-1/2)^2-(\theta_t+1/2)^2-\theta_0^2)}
\prod_{\varepsilon=\pm}G_q\big(1-\Lambda+\varepsilon \big(\sigma-\tfrac{1}{2}\big)\big)^{-1}.
\end{gather*}

\begin{Proposition}\label{prop_qPV_tau} Set
\begin{gather}
 \theta_1+\theta_\infty=\Lambda,\qquad \theta_1-\theta_\infty=\theta_*,\qquad t = q^{\Lambda}t_1,\nonumber\\
 s = \tilde s (q-1)^{-2\sigma}q^{-2\sigma\Lambda}
\prod_{\varepsilon=\pm}\Gamma_q\big(\tfrac{1}{2}-\Lambda+\varepsilon\sigma\big)^{-\varepsilon}.\label{eq_limit-V1}
\end{gather}
Then we have
\begin{alignat*}{3}
& C_i\tau_i^{\mathrm{VI}}(\theta_\infty,\theta_1,\theta_t,\theta_0\,|\, s,\sigma,t) \to\tau_i^{\mathrm{V}}(\theta_*,\theta_t,\theta_0\,|\, \tilde s, \sigma,t_1),\qquad && i=1,2,3,4,&\\
& C_5\tau_5^{\mathrm{VI}}(\theta_\infty,\theta_1,\theta_t,\theta_0\,|\, s,\sigma,t)\to \tau_1^{\mathrm{V}}(\theta_*,\theta_t,\theta_0\,|\, \tilde s, \sigma,qt_1),&&&\\
& C_6\tau_6^{\mathrm{VI}}(\theta_\infty,\theta_1,\theta_t,\theta_0\,|\, s,\sigma,t)\to \tau_2^{\mathrm{V}}(\theta_*,\theta_t,\theta_0\,|\, \tilde s, \sigma,t_1/q),&&&\\
& C_i\tau_i^{\mathrm{VI}}(\theta_\infty,\theta_1,\theta_t,\theta_0\,|\, s,\sigma,t)\to \tau_{i-2}^{\mathrm{V}}(\theta_*,\theta_t,\theta_0\,|\, \tilde s, \sigma,t_1),\qquad && i=7,8,&
\end{alignat*}
as $\Lambda\to\infty$. Here, we denote by $\tau_i^{\mathrm{VI}}(\theta_\infty,\theta_1,\theta_t,\theta_0\,|\, s,\sigma,t)$ the tau functions of $q$-$\mathrm{P_{VI}}$ presented in the previous section.
\end{Proposition}
\begin{proof} First, we verify the limit of the series part. For any partition $\lambda$ we have
\begin{gather*}
N_{\varnothing,\lambda}\big(q^{-\Lambda}u\big)q^{\Lambda|\lambda|} =\prod_{\square\in\lambda}
\big( q^{\Lambda}-q^{\ell_\lambda(\square)+a_\varnothing(\square)+1}u\big) \to f_{\lambda}\big(u^{-1}\big) ,\qquad \Lambda\to \infty.
\end{gather*}
Hence, the series $Z\left[\begin{matrix}
 \theta_1 & \theta_t\\
 \theta_\infty & \theta_0\\
\end{matrix}\Bigl|
\sigma,t\right]$ goes to
$Z_{\mathrm{V}}[\theta_*,\theta_t,\theta_0\,|\,\sigma,t]$
as $\Lambda\to \infty$.

Second, we examine the limits of the coefficients of $Z$. By the identities \eqref{eq_Gamma_Barnes} on $q$-Gamma function and $q$-Barnes function, for $n\in\mathbb{Z}$ we have
\begin{gather}
\prod_{\varepsilon=\pm}G_q(1-x+\varepsilon(\sigma+n))
= \prod_{\varepsilon=\pm}G_q(1-x+\varepsilon\sigma)
\Gamma_q(-x+\varepsilon\sigma)^{\varepsilon n}\prod_{i=0}^{|n|-1} \left[-x+\frac{|n|}{n}\sigma\right] \nonumber\\
\hphantom{\prod_{\varepsilon=\pm}G_q(1-x+\varepsilon(\sigma+n))=}{}
\times \prod_{i=0}^{|n|-1}\prod_{j=1}^{i}[-x+\sigma +j]\prod_{i=0}^{|n|-1}\prod_{j=1}^{i} [-x-\sigma -j].\label{eq_Gqn}
\end{gather}

Using the identity above, we compute the coefficient of $Z$ in $\tau_1^{\mathrm{VI}}$ multiplied by $C_1$ as follows
\begin{gather*}
 C_1s^n C\left[\begin{matrix}
\theta_1 & \theta_t \\
\theta_\infty + \tfrac{1}{2} & \theta_0 \\
\end{matrix}\Bigl|
\sigma+n\right]
t^{(\sigma+n)^2-\theta_t^2-\theta_0^2}
\\
\qquad{} =\tilde{s}^n (q-1)^{\sigma^2-2\sigma n}q^{-(\sigma^2-\theta_t^2-\theta_0^2)/2}
t_1^{(\sigma+n)^2-\theta_t^2-\theta_0^2} q^{\Lambda n^2}
\prod_{\varepsilon=\pm}\left(\frac{\Gamma_q\big({-}\Lambda-\tfrac{1}{2}+\varepsilon \sigma\big)}
{\Gamma_q\big({-}\Lambda+\tfrac{1}{2}+\varepsilon \sigma\big)} \right)^{\varepsilon n}
\\
\qquad\quad{} \times \prod_{i=0}^{|n|-1}\left[-\Lambda-\tfrac{1}{2}+\frac{|n|}{n}\sigma\right]
\prod_{i=0}^{|n|-1}\prod_{j=1}^{i}
[-\Lambda-{\tfrac{1}{2}}+\sigma +j]
\prod_{i=0}^{|n|-1}\prod_{j=1}^{i}
\big[{-}\Lambda-\tfrac{1}{2}-\sigma -j\big]
\\
\qquad\quad{} \times\frac{\prod\limits_{\varepsilon=\pm} G_q(1-\theta_*-{\tfrac{1}{2}}+\varepsilon(\sigma+n))\prod\limits_{\varepsilon,\varepsilon'=\pm}G_q(1+\varepsilon(\sigma+n)-\theta_t+\varepsilon'\theta_0)}
{G_q(1+2(\sigma+n))G_q(1-2(\sigma+n))}.
\end{gather*}
Then we have as $\Lambda\to \infty$ by the definition of $q$-number
\begin{gather*}
 q^{\Lambda n^2}\prod_{i=0}^{|n|-1}\big[{-}\Lambda-\tfrac{1}{2}+\frac{|n|}{n}\sigma\big]
\prod_{i=0}^{|n|-1}\prod_{j=1}^{i}\big[{-}\Lambda-{\tfrac{1}{2}}+\sigma +j\big]\prod_{i=0}^{|n|-1}\prod_{j=1}^{i}\big[{-}\Lambda-\tfrac{1}{2}-\sigma -j\big]
\\
\qquad{} \to (q-1)^{-n^2}\prod_{i=0}^{|n|-1}q^{-1/2+|n|\sigma/n}\prod_{i=0}^{|n|-1}\prod_{j=1}^{i}q^{-1/2+\sigma +j} \prod_{i=0}^{|n|-1}\prod_{j=1}^{i}q^{-1/2-\sigma -j}
\\
\qquad{} =(q-1)^{-n^2}q^{-n^2/2+\sigma n},
\end{gather*}
and by the identity \eqref{eq_Gamma_Barnes} of $q$-Gamma function
\begin{gather*}
\prod_{\varepsilon=\pm}\left(\frac{\Gamma_q(-\Lambda-\tfrac{1}{2}+\varepsilon \sigma)}
{\Gamma_q(-\Lambda+\tfrac{1}{2}+\varepsilon \sigma)} \right)^{\varepsilon n}
= \left(\frac{[-\Lambda-\tfrac{1}{2}-\sigma]}{[-\Lambda-\tfrac{1}{2}+\sigma]}\right)^n
\to q^{-2\sigma n}.
\end{gather*}
Therefore we obtain
\begin{gather*}
 C_1s^n C\left[\begin{matrix}
\theta_1 & \theta_t \\
\theta_\infty + \tfrac{1}{2} & \theta_0 \\
\end{matrix}\Bigl|
\sigma+n\right]
t^{(\sigma+n)^2-\theta_t^2-\theta_0^2}\\
\qquad{} \to \tilde{s}^n \big(t_1/\sqrt{q}\big)^{(\sigma+n)^2-\theta_t^2-\theta_0^2}
C_{\mathrm{V}}\big[\theta_*-\tfrac{1}{2},\theta_t,\theta_0\,|\,\sigma+n\big]
\end{gather*}
as $\Lambda \to \infty$. Similarly, we can compute the coefficients of $Z$ in the other tau functions and obtain the desired results.
\end{proof}

In what follows, we abbreviate $\tau_i^{\mathrm{V}}(\theta_*,\theta_t,\theta_0\,|\, s,\sigma,t)$ to $\tau_i$.
\begin{Theorem}\label{thm_qPV}
The functions
\begin{gather*}
 y=q^{-\theta_*-1}(q-1)^{1/2}t\frac{\tau_3\tau_4}{\tau_1\tau_2},\qquad
z=-\frac{\underline{\tau_1}\tau_2-\tau_1\underline{\tau_2}}{q^{\theta_*/2+1/2}\tau_1\underline{\tau_2}}
\end{gather*}
solves the $q$-Painlev\'e V equation
\begin{gather}
\frac{y \overline y}{a_3 a_4}=-\frac{(\overline z-b_1 t)(\overline z -b_2 t)}{\overline z-b_3},\qquad
\frac{z \overline z}{b_3}=-\frac{(y-a_1t)(y-a_2 t)}{a_4(y-a_3)}\label{eq_qPV}
\end{gather}
with the parameters
\begin{gather*}
 a_1=q^{-\theta_*-1},\qquad a_2=q^{-2\theta_t-\theta_*-1},\qquad a_3=q^{-1}, \qquad
a_4=q^{-3\theta_*/2-1/2},\\
b_1=q^{-\theta_0-\theta_t-\theta_*/2},\qquad b_2=q^{\theta_0-\theta_t-\theta_*/2},\qquad b_3=q^{-\theta_*/2-1/2}.
\end{gather*}
\end{Theorem}
\begin{proof}By definition we have
\begin{gather*}
C_1C_2=(q-1)^{1/2}C_3C_4.
\end{gather*}
Hence, by \eqref{eq_limit-V1} the solution ($y,z$) of the $q$-Painlev\'e VI equation has the following limit
\begin{gather*}
y \to y_1=q^{-\theta_*-1}(q-1)^{1/2}t_1\frac{\tau_3\tau_4}{\tau_1\tau_2},
\qquad
q^{-\Lambda/2}z\to z_1=-\frac{\underline{\tau_1}\tau_2-\tau_1\underline{\tau_2}}{q^{\theta_*/2+1/2}\tau_1\underline{\tau_2}},\qquad \Lambda\to \infty .
\end{gather*}
Substituting \eqref{eq_limit-V1} into the $q$-Painlev\'e VI equation \eqref{eq_qPVI},
we get
\begin{gather}
 \frac{y \overline y}{q^{-\Lambda-\theta_*-2}}
=\frac{\big(\overline z- q^{-\theta_0-\theta_t +(\Lambda-\theta_*)/2}t_1\big)
\big(\overline z -q^{\theta_0-\theta_t +(\Lambda-\theta_*)/2}t_1\big)}
{\big(\overline z-q^{(\Lambda-\theta_*-1)/2}\big)\big(\overline z-q^{-(\Lambda+\theta_*+1)/2)}\big)},\label{eq_qPVI_limit1}\\
 \frac{z \overline z}{q^{-1}}=-\frac{\big(y-q^{-\theta_*-1}t_1\big)
\big(y-q^{-2\theta_t-\theta_*-1}t_1\big)}{\big(y-q^{-1}\big)\big(y-q^{-\Lambda-\theta_*-1}\big)}.\label{eq_qPVI_limit2}
\end{gather}
Hence, since $y\to y_1$, $q^{-\Lambda/2}z\to z_1$ as $\Lambda\to \infty$, the system \eqref{eq_qPVI_limit1}, \eqref{eq_qPVI_limit2} degenerate to the $q$-Painle\-v\'e~V equation~\eqref{eq_qPV} for $y=y_1$ and $z=z_1$ as $\Lambda\to \infty$.
\end{proof}

Since we also have
\begin{gather*}
C_5C_6=(q-1)^{1/2}C_7C_8,\qquad C_1C_2=C_5C_6,
\end{gather*}
we obtain the following conjecture.
\begin{Conjecture}
The tau functions $\tau_i$ $(i=1,\dots, 6)$ satisfy the following bilinear equations
\begin{gather}
 \tau_1\tau_2-q^{-\theta_*}(q-1)^{1/2}t \tau_3\tau_4-\big(1-q^{-\theta_*}t\big)\overline{\tau_1}\underline{\tau_2}=0,\label{bilin-1V}\\
 (q-1)^{-1/2}\tau_1\tau_2-\tau_3\tau_4+\big(1-q^{-\theta_*}t\big)q^{2\theta_t}\underline{\tau_5}\overline{\tau_6}=0,\label{bilin-3V}\\
 (q-1)^{-1/2}\tau_1\tau_2-q^{2\theta_t}\tau_3\tau_4+q^{2\theta_t}\tau_5 \tau_6=0,\label{bilin-4V}\\
 \tau_1\underline{\tau_2}+q^{\theta_t-1/2}(q-1)^{1/2}t\underline{\tau_5}\tau_6-\underline{\tau_1}\tau_2=0,\label{bilin-5V}\\
 (q-1)^{-1/2}\tau_1\underline{\tau_2}+q^{\theta_0+2\theta_t}\underline{\tau_5}\tau_6-q^{\theta_t} \underline{\tau_3}\tau_4=0,\label{bilin-7V}\\
 (q-1)^{-1/2}\tau_1\underline{\tau_2}+q^{-\theta_0+2\theta_t}\underline{\tau_5}\tau_6-q^{\theta_t} \tau_3\underline{\tau_4}=0 .\label{bilin-8V}
\end{gather}
Then the functions
\begin{gather*}
 y=q^{-\theta_*-1}(q-1)^{1/2}t\frac{\tau_3\tau_4}{\tau_1\tau_2},\qquad z=-q^{\theta_t-\theta_*/2-1}(q-1)^{1/2}t\frac{\underline{\tau_5}\tau_6}{\tau_1\underline{\tau_2}}
\end{gather*}
solves $q$-$\mathrm{P_V}$ \eqref{eq_qPV}.
\end{Conjecture}

The four-term bilinear equation \eqref{eq_qPVI_4termbilinear} admits the following limit.
\begin{Proposition}\label{prop_qPV_twotermbilinear}We have
\begin{gather}
 \underline{\tau_1}\tau_2-\tau_1\underline{\tau_2}
=\frac{q^{-1/2}(q-1)^{1/2}}{q^{\theta_0}-q^{-\theta_0}}t
 (\underline{\tau_3}\tau_4-\tau_3\underline{\tau_4} ).\label{bilin-9V}
\end{gather}
\end{Proposition}
\begin{proof} The identity \eqref{bilin-9V} is a direct consequence of \eqref{eq_qPVI_4termbilinear} by the limit \eqref{eq_limit-V1} as $\Lambda\to \infty$.
\end{proof}

We remark that tau functions without the Chern--Simons term is also obtained by the limit
\begin{gather*}
 \theta_1+\theta_\infty=-\Lambda,\qquad \theta_1-\theta_\infty=\theta_*, \qquad s = \tilde s (q-1)^{-2\sigma}
\prod_{\varepsilon=\pm}\Gamma_q\big(\tfrac{1}{2}+\Lambda+\varepsilon\sigma\big)^{-\varepsilon},\qquad \Lambda\to\infty
\end{gather*}
from the tau functions of $q$-$\mathrm{P_{VI}}$.

\section[From $q$-$\mathrm{P_V}$ to $q$-$\mathrm{P_{{III}_1}}$]{From $\boldsymbol{q}$-$\boldsymbol{\mathrm{P_V}}$ to $\boldsymbol{q}$-$\boldsymbol{\mathrm{P_{{III}_1}}}$}\label{section4}

In this section, we take a limit of the tau functions of $q$-$\mathrm{P_V}$ to $q$-$\mathrm{P_{{III}_1}}$. Define the tau function by
\begin{gather*}
\tau^{\mathrm{{III}_1}}(\theta_*,\theta_\star\,|\, s,\sigma,t)= \sum_{n\in\mathbb{Z}}
s^n t^{(\sigma+n)^2}C_{\mathrm{{III}_1}} [\theta_*,\theta_\star\,|\,\sigma+n ]
 Z_{\mathrm{{III}_1}} [\theta_*,\theta_\star\,|\,\sigma+n,t ] ,
\end{gather*}
with
\begin{gather*}
C_{\mathrm{{III}_1}} [\theta_*,\theta_\star\,|\,\sigma ]= (q-1)^{-2\sigma^2}
\prod_{\varepsilon=\pm} \frac{G_q(1-\theta_*+\varepsilon\sigma)
G_q(1+\varepsilon\sigma-\theta_\star)}{G_q(1+2\varepsilon\sigma)},\\
Z_{\mathrm{{III}_1}} [\theta_*,\theta_\star\,|\,\sigma,t ]=
\sum_{(\lambda_+,\lambda_-)\in\mathbb{Y}^2}
t^{|\lambda_+|+|\lambda_-|}
\frac{\prod\limits_{\varepsilon=\pm}N_{\varnothing,\lambda_{\varepsilon}}
\big(q^{-\theta_*-\varepsilon \sigma}\big)
N_{\lambda_{\varepsilon},\varnothing}\big(q^{\varepsilon\sigma-\theta_\star}\big)}
{\prod\limits_{\varepsilon,\varepsilon'=\pm}N_{\lambda_\varepsilon,\lambda_{\varepsilon'}}\big(q^{(\varepsilon-\varepsilon')\sigma}\big)} .
\end{gather*}

Let us define the tau functions for $q$-$\mathrm{P_{III_1}}$ by
\begin{alignat*}{3}
& \tau^{\mathrm{{III}_1}}_1= \tau^{\mathrm{{III}_1}}\big(\theta_*-\tfrac{1}{2},\theta_\star\,|\, s,\sigma,t/\sqrt{q}\big),\qquad&&
\tau^{\mathrm{{III}_1}}_2=\tau^{\mathrm{{III}_1}}\big(\theta_*+\tfrac{1}{2},\theta_\star\,|\, s,\sigma,\sqrt{q}t\big),&\\
& \tau^{\mathrm{{III}_1}}_3= \tau^{\mathrm{{III}_1}}\big(\theta_*,\theta_\star-\tfrac{1}{2}\,|\, s,\sigma+\tfrac{1}{2},t/\sqrt{q}\big),\qquad&&
\tau^{\mathrm{{III}_1}}_4=\tau^{\mathrm{{III}_1}}\big(\theta_*,\theta_\star+\tfrac{1}{2}\,|\, s,\sigma-\tfrac{1}{2},\sqrt{q}t\big).&
\end{alignat*}
Put
\begin{gather*}
C_1=(q-1)^{-\sigma^2}q^{-\Lambda\sigma^2-(\theta_t^2+\theta_0^2)/2}t^{\theta_t^2+\theta_0^2}\prod_{\varepsilon=\pm}G_q (1-\Lambda+\varepsilon \sigma )^{-1},\\
C_2= (q-1)^{-\sigma^2}q^{-\Lambda\sigma^2+(\theta_t^2+\theta_0^2)/2}t^{\theta_t^2+\theta_0^2}\prod_{\varepsilon=\pm}G_q(1-\Lambda+\varepsilon \sigma)^{-1},\\
C_3=(q-1)^{-(\sigma+1/2)^2}q^{-(\Lambda+1/2)(\sigma+1/2)^2}t^{\theta_t^2+(\theta_0+1/2)^2}\prod_{\varepsilon=\pm}G_q\big(\tfrac{1}{2}-\Lambda+\varepsilon\big( \sigma+\tfrac{1}{2}\big)\big)^{-1},\\
C_4=(q-1)^{-(\sigma-1/2)^2}q^{-(\Lambda-1/2)(\sigma-1/2)^2}t^{\theta_t^2+(\theta_0-1/2)^2}
\prod_{\varepsilon=\pm}G_q\big( \tfrac{3}{2} -\Lambda+\varepsilon\big( \sigma-\tfrac{1}{2}\big)\big)^{-1},\\
C_5=(q-1)^{-(\sigma+1/2)^2}q^{-(\Lambda-1/2)(\sigma+1/2)^2}t^{(\theta_t-1/2)^2+\theta_0^2}
\prod_{\varepsilon=\pm}G_q\big( \tfrac{3}{2} -\Lambda+\varepsilon\big( \sigma+\tfrac{1}{2}\big)\big)^{-1},\\
C_6=(q-1)^{-(\sigma-1/2)^2}q^{-(\Lambda+1/2)(\sigma-1/2)^2} t^{(\theta_t+1/2)^2+\theta_0^2}
\prod_{\varepsilon=\pm}G_q\big(\tfrac{1}{2}-\Lambda+\varepsilon\big( \sigma-\tfrac{1}{2}\big)\big)^{-1}.
\end{gather*}

\begin{Proposition}Set
\begin{gather}
\theta_t+\theta_0=\Lambda,\qquad \theta_t-\theta_0=\theta_\star,\qquad t = q^{\Lambda}t_1,\nonumber\\
s = \tilde s (q-1)^{-2\sigma}q^{-\sigma(2\Lambda+1)} \prod_{\varepsilon=\pm}\Gamma_q(-\Lambda+\varepsilon\sigma)^{-\varepsilon} .\label{eq_limit-III1}
\end{gather}
Then we have
\begin{alignat*}{3}
& C_i\tau_i^{\mathrm{V}}(\theta_*,\theta_t,\theta_0\,|\, s,\sigma,t)\to \tau_i^{\mathrm{{III}_1}}(\theta_*,\theta_\star \,|\, \tilde s, \sigma,t_1), \qquad && i=1,2,3,4,& \\
& C_i\tau_{i}^{\mathrm{V}}(\theta_*,\theta_t,\theta_0\,|\, s,\sigma,t)\to\tau_{i-2}^{\mathrm{{III}_1}}(\theta_*,\theta_\star \,|\, \tilde s, \sigma,qt_1), \qquad && i=5,6,&
\end{alignat*}
as $\Lambda\to \infty$.
\end{Proposition}
\begin{proof}For any partition $\lambda$ we have
\begin{gather*}
N_{\lambda,\varnothing}\big(q^{-\Lambda}u\big)q^{\Lambda|\lambda|} =\prod_{\square\in\lambda}
\big( q^{\Lambda}-q^{-\ell_\lambda(\square)-a_\varnothing(\square)-1}u\big)
\to f_{\lambda}(u)^{-1}, \qquad \Lambda\to \infty.
\end{gather*}
Hence, the series $ Z_{\mathrm{V}} [ \theta_*,\theta_t,\theta_0\,|\, \sigma, t]$ goes to $Z_{\mathrm{III_1}} [\theta_*,\theta_\star\,|\, \sigma,t_1 ]$ as $\Lambda\to \infty$.
The coefficients of~$Z_{\mathrm{V}}$ are computed in the same way as in the proof of Proposition~\ref{prop_qPV_tau} using~\eqref{eq_Gqn} and we obtain the desired results.
\end{proof}

In what follows, we abbreviate $\tau_i^{\mathrm{III_1}}(\theta_*,\theta_\star\,|\, s,\sigma,t)$ to $\tau_i$. Fortunately, the four-term bilinear equation \eqref{bilin-9V} degenerates to a three-term bilinear equation.
\begin{Proposition}We have
\begin{gather}
\underline{\tau_1}\tau_2-\tau_1\underline{\tau_2}=q^{-1/4}t^{1/2} \tau_3\underline{\tau_4}.\label{bilin-9III_1}
\end{gather}
\end{Proposition}
\begin{proof}By definition and \eqref{eq_limit-III1} we have
\begin{gather*}
\underline{C_1}C_2= C_1\underline{C_2} =\big(q^{-\Lambda}-q^{\sigma}\big)(q-1)^{-1/2}t_1^{-1/2}q^{\theta_0+1/4}\underline{C_3}C_4\\
\hphantom{\underline{C_1}C_2= C_1\underline{C_2}}{} =\big(q^{-\Lambda}-q^{\sigma}\big)(q-1)^{-1/2}t_1^{-1/2}q^{-\theta_0+1/4}C_3\underline{C_4}.
\end{gather*}
Hence from the four-term bilinear equation \eqref{bilin-9V} degenerates to the three-term bilinear equa\-tion~\eqref{bilin-9III_1} by \eqref{eq_limit-III1} as $\Lambda\to \infty$.
\end{proof}

\begin{Theorem}The functions
\begin{gather}\label{eq_qPIII_1yz}
y=q^{-\theta_*-1} t^{1/2}\frac{\tau_3\tau_4}{\tau_1\tau_2},\qquad
z=q^{-\theta_*/2-3/4}t^{1/2}
\frac{\tau_3\underline{\tau_4}}
{\underline{\tau_1}\tau_2}
\end{gather}
solves the $q$-Painlev\'e $\mathrm{III}_1$ equation
\begin{gather}
\frac{y \overline y}{a_3 a_4}=-\frac{\overline z (\overline z -b_2 t)}{\overline z-b_3},\qquad
\frac{z \overline z}{b_3}=-\frac{y(y-a_2 t)}{a_4(y-a_3)}\label{eq_qPIII_1}
\end{gather}
with the parameters
\begin{gather*}
 a_2=q^{-\theta_\star-\theta_*-1},\qquad\! a_3=q^{-1}, \qquad\!
a_4=q^{-3\theta_*/2-1/2}, \qquad \!
 b_2=q^{-\theta_*/2},\qquad\!
b_3=q^{-\theta_*/2-1/2}.
\end{gather*}
Furthermore, the tau functions $\tau_i$ $(i=1,\dots, 4)$ satisfy the following bilinear equations.
\begin{gather}
\tau_1\tau_2- q^{-\theta_*} t^{1/2} \tau_3\tau_4-\overline{\tau_1} \underline{\tau_2}=0,\label{bilin-1III_1}\\
\tau_1\tau_2-q^{\theta_\star}t^{-1/2} \tau_3\tau_4+q^{\theta_\star}t^{-1/2} \overline{\tau_3}\underline{ \tau_4}=0,\label{bilin-4III_1}\\
\tau_1\underline{\tau_2}+q^{-1/4} t^{1/2}\tau_3\underline{\tau_4}-\underline{\tau_1}\tau_2=0,\label{bilin-5III_1}\\
\tau_1\underline{\tau_2}+q^{1/4}t^{-1/2}\tau_3\underline{\tau_4}-q^{1/4}t^{-1/2}\underline{\tau_3}\tau_4=0.\label{bilin-7III_1}
\end{gather}
\end{Theorem}
\begin{proof}By definition and \eqref{eq_limit-III1} we have
\begin{gather*}
C_1C_2=\big(q^{-\Lambda}-q^{\sigma}\big)(q-1)^{-1/2}t_1^{-1/2}C_3C_4.
\end{gather*}
Hence, by \eqref{eq_limit-III1} and \eqref{bilin-9III_1} the solution ($y,z$) of the $q$-Painlev\'e V equation degenerates to
\begin{gather*}
y \to y_1=q^{-\theta_*-1} t_1^{1/2}\frac{\tau_3\tau_4}{\tau_1\tau_2},\qquad z \to z_1=q^{-\theta_*/2-3/4}t_1^{1/2} \frac{\tau_3\underline{\tau_4}}
{\underline{\tau_1}\tau_2},\qquad \Lambda\to \infty .
\end{gather*}
Also, the $q$-Painlev\'e V equation \eqref{eq_qPV} degenerates to the $q$-Painlev\'e $\mathrm{III}_1$ equation~\eqref{eq_qPIII_1} for $y=y_1 $ and $z=z_1$ as $\Lambda\to \infty$.

Next we prove the bilinear equations \eqref{bilin-1III_1}--\eqref{bilin-7III_1}. The bilinear equation \eqref{bilin-5III_1} is \eqref{bilin-9III_1}. The identity~\eqref{bilin-7III_1} is obtained by substituting the expression~\eqref{eq_qPIII_1yz} of $(y,z)$ into the $q$-Painlev\'e $\mathrm{III_1}$ equation{\samepage
 \begin{gather*}
 \frac{y \overline y}{a_3 a_4}=-\frac{\overline z (\overline z -b_2 t)}{\overline z-b_3},
 \end{gather*}
and using the bilinear equation \eqref{bilin-5III_1}.}

In order to prove \eqref{bilin-1III_1} and \eqref{bilin-4III_1}, we use the following transformation
 \begin{gather}
\big( \tilde{\theta}_*, \tilde{\theta}_\star, \tilde{\sigma}, \tilde{s}, \tilde{t}\big)= \big( -\theta_\star, -\theta_*, \sigma-\tfrac{1}{2}, Cs,
q^{-\theta_*-\theta_\star+1/2}t\big),\label{eq_qPIII_1_transformation}
\end{gather}
where
\begin{gather*}
C= q^{(\sigma-1)(2\theta_*+2\theta_\star+1)}
\prod_{\varepsilon,\varepsilon'=\pm}\Gamma_q\big( \tfrac{1}{2}+\varepsilon\theta_*+\varepsilon'(\sigma-1)\big)^{-\varepsilon\varepsilon'}
\Gamma_q\big( \tfrac{1}{2}+\varepsilon\big(\theta_\star+\tfrac{1}{2}\big)+\varepsilon'(\sigma-1)\big)^{-\varepsilon\varepsilon'}.
\end{gather*}
From the definition of the Nekrasov factor, for a partition $\lambda$ we have
\begin{gather*}
N_{\varnothing,\lambda}(u)N_{\lambda,\varnothing}(w)= (uw)^{|\lambda|}N_{\varnothing,\lambda}\big(w^{-1}\big)N_{\lambda,\varnothing}\big(u^{-1}\big).
\end{gather*}
By the identity above, the series part $Z$ of the tau functions $\tau_1,\dots,\tau_4$ transform to
\begin{gather*}
Z_{\mathrm{{III}_1}}\big[\tilde\theta_*-\tfrac{1}{2},\tilde\theta_\star\,|\,\tilde\sigma,\tilde t/\sqrt{q}\big]
= Z_{\mathrm{{III}_1}}\big[\theta_*,\theta_\star+\tfrac{1}{2}\,|\,\sigma-\tfrac{1}{2},\sqrt{q} t\big],\\
Z_{\mathrm{{III}_1}}\big[\tilde\theta_*+\tfrac{1}{2},\tilde\theta_\star\,|\,\tilde\sigma,\sqrt{q}\tilde t\big]
= Z_{\mathrm{{III}_1}}\big[\theta_*,\theta_\star-\tfrac{1}{2}\,|\,\sigma-\tfrac{1}{2},\sqrt{q} t\big],\\
Z_{\mathrm{{III}_1}}\big[\tilde\theta_*,\tilde\theta_\star-\tfrac{1}{2}\,|\,\tilde\sigma+\tfrac{1}{2},\tilde t/\sqrt{q}\big]
= Z_{\mathrm{{III}_1}}\big[\theta_*+\tfrac{1}{2},\theta_\star\,|\,\sigma,\sqrt{q} t\big],\\
Z_{\mathrm{{III}_1}}\big[\tilde\theta_*,\tilde\theta_\star+\tfrac{1}{2}\,|\,\tilde\sigma-\tfrac{1}{2},\tilde t\big]
= Z_{\mathrm{{III}_1}}\big[\theta_*-\tfrac{1}{2},\theta_\star\,|\,\sigma-1,\sqrt{q} t\big],
\end{gather*}
respectively. Using the identity
\begin{gather*}
\frac{G_q(1+x+n)G_q(1-x)}{G_q(1-x-n)G_q(1+x)} =(-1)^{n(n+1)/2}q^{n(n+1)x/2+(n-1)n(n+1)/6}\Gamma_q(x)^n\Gamma_q(1-x)^n
\end{gather*}
for $n\in\mathbb{Z}$, we can compute the coefficients $C_{\mathrm{III_1}}$
and obtain
\begin{alignat*}{3}
&\tilde \tau_1=K\big[\theta_*,\theta_\star+\tfrac{1}{2},\sigma-\tfrac{1}{2}\big]\tau_4,\qquad&&
\tilde \tau_2=s K\big[\theta_*,\theta_\star-\tfrac{1}{2},\sigma-\tfrac{1}{2}\big] \overline{\tau_3},&\\
&\tilde \tau_3=K\big[\theta_*+\tfrac{1}{2},\theta_\star,\sigma\big]\tau_2,\qquad &&
\tilde \tau_4=s K\big[\theta_*-\tfrac{1}{2},\theta_\star,\sigma-1\big] \overline{\tau_1},&
\end{alignat*}
where we denote by $\tilde{\tau}_i$ the tau functions with parameters $ ( \tilde{\theta}_*, \tilde{\theta}_\star, \tilde{\sigma}, \tilde{s}, \tilde{t} )$ and by $\tau_i$ the tau functions with parameters $(\theta_*,\theta_\star,\sigma,s,t)$, and
\begin{gather*}
K [\theta_*,\theta_\star,\sigma ]= q^{-(\theta_*+\theta_\star)\sigma^2}
\prod_{\varepsilon,\varepsilon'=\pm}G_q(1+\varepsilon \theta_*+\varepsilon'\sigma)^\varepsilon G_q(1+\varepsilon\theta_\star+\varepsilon'\sigma)^\varepsilon .
\end{gather*}
By definition we have
\begin{gather}
 \frac{K\big[\theta_*,\theta_\star+\tfrac{1}{2},\sigma-\tfrac{1}{2}\big]
K\big[\theta_*,\theta_\star-\tfrac{1}{2},\sigma-\tfrac{1}{2}\big]}{K\big[\theta_*+\tfrac{1}{2},\theta_\star,\sigma\big]K\big[\theta_*-\tfrac{1}{2},\theta_\star,\sigma-1\big]}=
-q^{(\theta_\star-\theta_*)/2} . \label{eq_K1}
\end{gather}
Applying the transformation \eqref{eq_qPIII_1_transformation} to the bilinear equations \eqref{bilin-5III_1} and \eqref{bilin-7III_1} and using the rela\-tion~\eqref{eq_K1}, we obtain the identities \eqref{bilin-1III_1} and \eqref{bilin-4III_1}.
\end{proof}

We note that the bilinear equations \eqref{bilin-1V}, \eqref{bilin-4V}, \eqref{bilin-5V}, and \eqref{bilin-7V} for the tau functions of~$q$-$\mathrm{P_V}$ degenerate to \eqref{bilin-1III_1}, \eqref{bilin-4III_1}, \eqref{bilin-5III_1}, and \eqref{bilin-7III_1}, respectively.

\section[From $q$-$\mathrm{P_{{III}_1}}$ to $q$-$\mathrm{P_{{III}_2}}$]{From $\boldsymbol{q}$-$\boldsymbol{\mathrm{P_{{III}_1}}}$ to $\boldsymbol{q}$-$\boldsymbol{\mathrm{P_{{III}_2}}}$}\label{section5}

In this section, we take a limit of the tau functions of $q$-$\mathrm{P_{III_1}}$ to $q$-$\mathrm{P_{III_2}}$. Define the tau function by
\begin{gather*}
\tau^{\mathrm{{III}_2}}(\theta_*\,|\, s,\sigma,t)= \sum_{n\in\mathbb{Z}}s^n t^{(\sigma+n)^2} C_{\mathrm{{III}_2}} [\theta_*\,|\,\sigma+n ]
 Z_{\mathrm{{III}_2}} [\theta_*\,|\,\sigma+n,t ] ,
\end{gather*}
with
\begin{gather*}
C_{\mathrm{{III}_2}} [\theta_*\,|\,\sigma ] = (q-1)^{-3\sigma^2}
\prod_{\varepsilon=\pm} \frac{G_q(1-\theta_*+\varepsilon\sigma)}{G_q(1+2\varepsilon\sigma)},\\
Z_{\mathrm{{III_2}}} [\theta_*\,|\,\sigma,t ]
= \sum_{(\lambda_+,\lambda_-)\in\mathbb{Y}^2}
t^{|\lambda_+|+|\lambda_-|}
\frac{\prod\limits_{\varepsilon=\pm}N_{\varnothing,\lambda_{\varepsilon}}
\big(q^{-\theta_*-\varepsilon \sigma}\big)
f_{\lambda_{\varepsilon}}(q^{\varepsilon\sigma})^{-1}}
{\prod\limits_{\varepsilon,\varepsilon'=\pm}N_{\lambda_\varepsilon,\lambda_{\varepsilon'}}\big(q^{(\varepsilon-\varepsilon')\sigma}\big)} .
\end{gather*}

In the same way as in Section~\ref{section3}, it is possible to remove $f_{\lambda_\varepsilon}(q^{\varepsilon\sigma})^{-1}$ from $Z_\mathrm{{III}_2} [\theta_*\,|\,\sigma,t ]$ by change of variables. Because if we set
\begin{gather*}
Z_\mathrm{{III_2}}^{CS=0} [\theta_*\,|\,\sigma,t ] =\sum_{(\lambda_+,\lambda_-)\in\mathbb{Y}^2}
t^{|\lambda_+|+|\lambda_-|} \frac{\prod\limits_{\varepsilon=\pm}N_{\varnothing,\lambda_{\varepsilon}}\big(q^{-\theta_*-\varepsilon \sigma}\big)}
{\prod\limits_{\varepsilon,\varepsilon'=\pm}N_{\lambda_\varepsilon,\lambda_{\varepsilon'}}\big(q^{(\varepsilon-\varepsilon')\sigma}\big)},
\end{gather*}
then we have
\begin{gather*}
Z_\mathrm{III_2} [\theta_*\,|\,\sigma,t ] =Z_\mathrm{III_2}^{CS=0} \big[{-}\theta_*\,|\,\sigma,q^{-\theta_*}t \big].
\end{gather*}

Let us define the tau functions for $q$-$\mathrm{P_{III_2}}$ by
\begin{gather*}
\tau^{\mathrm{{III}_2}}_1= \tau^{\mathrm{{III}_2}}\big(\theta_*-\tfrac{1}{2}\,|\, s,\sigma,t/\sqrt{q}\big),\qquad
\tau^{\mathrm{{III}_2}}_2= \tau^{\mathrm{{III}_2}}\big(\theta_*+\tfrac{1}{2}\,|\, s,\sigma+1,\sqrt{q}t\big),\\
\tau^{\mathrm{{III}_2}}_3=\tau^{\mathrm{{III}_2}}\big(\theta_*\,|\, s,\sigma+\tfrac{1}{2},t\big).
\end{gather*}
Put
\begin{gather*}
C_1=(q-1)^{-\sigma^2}q^{-\Lambda\sigma^2}\prod_{\varepsilon=\pm}G_q (1-\Lambda+\varepsilon\sigma )^{-1},\\
C_2=C_1,\\
C_3=(q-1)^{-(\sigma+1/2)^2}q^{-(\Lambda-1/2)(\sigma+1/2)^2}\prod_{\varepsilon=\pm}G_q\big( \tfrac{3}{2} -\Lambda+\varepsilon\big(\sigma+\tfrac{1}{2}\big)\big)^{-1},\\
C_4=(q-1)^{-(\sigma-1/2)^2}q^{-(\Lambda+1/2)(\sigma-1/2)^2}\prod_{\varepsilon=\pm}G_q\big(\tfrac{1}{2}-\Lambda+\varepsilon\big(\sigma-\tfrac{1}{2}\big)\big)^{-1}.
\end{gather*}

\begin{Proposition}Set
\begin{gather*}
 \theta_\star=\Lambda,\qquad t = q^{\Lambda}t_1,\qquad s = \tilde s (q-1)^{-2\sigma}q^{-\sigma(2\Lambda+1)}
\prod_{\varepsilon=\pm}\Gamma_q (-\Lambda+\varepsilon\sigma )^{-\varepsilon}.
\end{gather*}
Then we have
\begin{gather*}
C_i\tau^{\mathrm{III_1}}_i(\theta_*,\theta_\star\,|\, s,\sigma,t)\to \tau^{\mathrm{III_2}}_i(\theta_*\,|\, \tilde s,\sigma,t_1),\qquad i=1,3,\\
C_2\tau^{\mathrm{III_1}}_2(\theta_*,\theta_\star\,|\, s,\sigma,t)\to \tilde s\tau^{\mathrm{III_2}}_2(\theta_*\,|\, \tilde s,\sigma,t_1),\\
C_4\tau^{\mathrm{III_1}}_4(\theta_*,\theta_\star\,|\, s,\sigma,t)\to \tilde s\tau^{\mathrm{III_2}}_3(\theta_*\,|\, \tilde s,\sigma,t_1)
\end{gather*}
as $\Lambda \to \infty$.
\end{Proposition}

In what follows, we abbreviate $\tau_i^{\mathrm{III_2}}(\theta_*\,|\, s,\sigma,t)$ to $\tau_i$. Since we have the relation
\begin{gather*}
 C_1C_2=(q-1)^{-1/2}\big(q^{-\Lambda/2}-q^{\Lambda/2-\sigma}\big)C_3C_4,
\end{gather*}
we obtain the following theorem by the degeneration.
\begin{Theorem}
The functions
\begin{gather*}
y=q^{-\theta_*-1}(q-1)^{-1/2} t^{1/2}\frac{\tau_3^2}{\tau_1\tau_2},\qquad
z=q^{-\theta_*/2-3/4}(q-1)^{-1/2}t^{1/2}\frac{\tau_3\underline{\tau_3}}{\underline{\tau_1}\tau_2}
\end{gather*}
solves the $q$-Painlev\'e $\mathrm{III_2}$ equation
\begin{gather*}
\frac{y \overline y}{a_3 a_4}=-\frac{\overline z^2}{\overline z-b_3} ,\qquad
\frac{z \overline z}{b_3}=-\frac{y(y-a_2 t)}{a_4(y-a_3)}
\end{gather*}
with the parameters
\begin{gather*}
 a_2=q^{-\theta_*-1},\qquad a_3=q^{-1}, \qquad a_4=q^{-3\theta_*/2-1/2},
\qquad b_2=q^{-\theta_*/2},\qquad b_3=q^{-\theta_*/2-1/2}.
\end{gather*}
 Furthermore, the tau functions $\tau_i$ ($i=1,2,3$) satisfy the following bilinear equations.
\begin{gather}
\tau_1\tau_2- q^{-\theta_*}(q-1)^{-1/2} t^{1/2} \tau_3^2-\overline{\tau_1} \underline{\tau_2}=0,\label{bilin-1III_2}\\
\tau_1\tau_2-(q-1)^{-1/2} t^{-1/2} \tau_3^2+(q-1)^{-1/2} t^{-1/2} \overline{\tau_3}\underline{ \tau_3}=0,\label{bilin-4III_2}\\
\tau_1\underline{\tau_2}+q^{-1/4}(q-1)^{-1/2}t^{1/2} \tau_3\underline{\tau_3}-\underline{\tau_1}\tau_2=0.\label{bilin-5III_2}
\end{gather}
\end{Theorem}

We note that the bilinear equations \eqref{bilin-1III_1}, \eqref{bilin-4III_1}, and \eqref{bilin-5III_1} for the tau functions of $q$-$\mathrm{P_{III_1}}$ degenerate to \eqref{bilin-1III_2}, \eqref{bilin-4III_2}, and \eqref{bilin-5III_2}, respectively.

\section[From $q$-$\mathrm{P_{{III}_2}}$ to $q$-$\mathrm{P_{{III}_3}}$]{From $\boldsymbol{q}$-$\boldsymbol{\mathrm{P_{{III}_2}}}$ to $\boldsymbol{q}$-$\boldsymbol{\mathrm{P_{{III}_3}}}$}\label{section6}

In this section, we take a limit of the tau functions of $q$-$\mathrm{P_{III_2}}$ to $q$-$\mathrm{P_{{III}_3}}$. Define the tau function by
\begin{gather*}
\tau^{\mathrm{III_3}}(s,\sigma,t)= \sum_{n\in\mathbb{Z}}s^n t^{(\sigma+n)^2}C_{\mathrm{{III}_3}}[\sigma+n] Z_{\mathrm{{III}_3}}[\sigma+n,t],
\end{gather*}
with
\begin{gather*}
C_{\mathrm{{III}_3}} [\sigma ]=(q-1)^{-4\sigma^2}\prod_{\varepsilon=\pm} \frac{1}{G_q(1+2\varepsilon (\sigma+n))},\\
Z_{\mathrm{{III_3}}}[\sigma,t]=\sum_{(\lambda_+,\lambda_-)\in\mathbb{Y}^2}t^{|\lambda_+|+|\lambda_-|}
\frac{1}{\prod\limits_{\varepsilon,\varepsilon'=\pm}N_{\lambda_\varepsilon,\lambda_{\varepsilon'}}\big(q^{(\varepsilon-\varepsilon')\sigma}\big)}.
\end{gather*}

Let us define the tau functions for $q$-$\mathrm{P_{III_3}}$ by
\begin{gather*}
\tau^{\mathrm{III_3}}_1=\tau^{\mathrm{III_3}}(s,\sigma,t),\qquad
\tau^{\mathrm{III_3}}_2=\tau^{\mathrm{III_3}}\big(s,\sigma+\tfrac{1}{2},t\big).
\end{gather*}
Put
\begin{gather*}
C_1=(q-1)^{-\sigma^2}q^{-(\Lambda-1/2)\sigma^2} \prod_{\varepsilon=\pm}G_q\big( \tfrac{3}{2}-\Lambda+\varepsilon\sigma\big)^{-1},\\
C_2=(q-1)^{-(\sigma+1)^2}q^{-(\Lambda+1/2)(\sigma+1)^2}\prod_{\varepsilon=\pm}G_q\big(\tfrac{1}{2}-\Lambda+\varepsilon(\sigma+1)\big)^{-1},\\
C_3=(q-1)^{-(\sigma+1/2)^2}q^{-\Lambda(\sigma+1/2)^2}\prod_{\varepsilon=\pm}G_q\big(1-\Lambda+\varepsilon\big(\sigma+\tfrac{1}{2}\big)\big)^{-1}.
\end{gather*}

\begin{Proposition}Set
\begin{gather*}
\theta_*=\Lambda,\qquad t = q^{\Lambda}t_1, \qquad s = \tilde s (q-1)^{-2\sigma}q^{-2\sigma\Lambda}
\prod_{\varepsilon=\pm}\Gamma_q\big(\tfrac{1}{2}-\Lambda+\varepsilon\sigma\big)^{-\varepsilon}.
\end{gather*}
Then we have
\begin{gather*}
C_1\tau_1^{\mathrm{III_2}}(\theta_*\,|\, s,\sigma,t) \to \tau_1^{\mathrm{III_3}}(\tilde s, \sigma, t_1),\\
C_2\tau_2^{\mathrm{III_2}}(\theta_*\,|\, s,\sigma,t)\to \tau_1^{\mathrm{III_3}}(\tilde s, \sigma, t_1)/\tilde s,\\
C_3\tau_3^{\mathrm{III_2}}(\theta_*\,|\, s,\sigma,t)\to \tau_2^{\mathrm{III_3}}(\tilde s, \sigma, t_1),
\end{gather*}
as $\Lambda\to\infty$.
\end{Proposition}

In what follows, we abbreviate $\tau_i^{\mathrm{III_3}}( s,\sigma,t)$ to $\tau_i$.
Since we have the relation
\begin{gather*}
 C_1C_2=(q-1)^{1/2}\frac{q^{-\sigma-1/2+\Lambda/2}}{q^{-\sigma-1/2}-q^\Lambda}C_3^2,
\end{gather*}
we obtain the following theorem by the degeneration.
\begin{Theorem}
The functions
\begin{gather*}
y=t^{1/2}\frac{s\tau_2^2}{\tau_1^2},\qquad z=q^{-3/4}t^{1/2}\frac{s\tau_2\underline{\tau_2}}{\tau_1\underline{\tau_1}}
\end{gather*}
solves the $q$-Painlev\'e $\mathrm{III}_3$ equation
\begin{gather}
\frac{y \overline y}{a_3 }=\overline z^2 ,\qquad
z \overline z=-\frac{y(y-a_2 t)}{y-a_3}\label{eq_qPIII_3}
\end{gather}
with the parameters
\begin{gather*}
 a_2=q^{-1},\qquad a_3=q^{-1}.
\end{gather*}
 Furthermore, the tau functions $\tau_1$, $\tau_2$ satisfy the following bilinear equations.
\begin{gather}
 st^{1/2} \tau_2^2-\tau_1^2 +\overline{\tau_1}\underline{\tau_1}=0,\label{bilin-1III_3}\\
s^{-1}t^{1/2}\tau_1^2-\tau_2^2+\overline{\tau_2}\underline{\tau_2}=0.\label{bilin-4III_3}
\end{gather}
\end{Theorem}
We note that the bilinear equations \eqref{bilin-1III_2}, \eqref{bilin-4III_2} for the tau functions of~$q$-$\mathrm{P_{III_2}}$ degenerate to~\eqref{bilin-1III_3}, \eqref{bilin-4III_3}, respectively. As suggested in \cite[equations~(2.9)--(2.11)]{BS}, the bilinear equa\-tion~\eqref{bilin-4III_3} is derived from~\eqref{bilin-1III_3} by the transformation $\sigma\to \sigma+1/2$.

\begin{Remark}The tau function $\mathcal{T}_c\big(q^{2\sigma},s;q\,|\,t\big)$ proposed in~\cite{BS} for the $q$-Painlev\'e $\mathrm{III}_3$ equation are related to our tau functions by
\begin{gather*}
\mathcal{T}_c\big(q^{2\sigma},s;q\,|\,t\big)= (-1)^{-2\sigma^2} \tau^{\mathrm{III_3}}\big( (-1)^{-4\sigma}s, \sigma, t\big).
\end{gather*}
\end{Remark}

\begin{Remark}$q$-$P(A_7')$ in \cite{Mu} (or $q$-$P\big(A_1^{(1)}/A_7^{(1)}\big)$ in \cite[equation~(8.14)]{KNY}) is
\begin{gather*}
\frac{y\overline y}{a_4}=-\frac{\overline z(\overline z-b_2 t)}{\overline z-b_3},\qquad \frac{z\overline z}{b_3}=\frac{y^2}{a_4},
\end{gather*}
where $y=y(t)$, $z=z(t)$, and $a_4$, $b_1$, $b_2$, $b_3$ are complex parameters. Replacing $y$, $z$ in \eqref{eq_qPIII_3} by~$z$,~$\underline{y}$, we obtain $q$-$P(A_7')$ with $a_4=1$, $b_2=1$, and $b_3=q^{-1}$.
\end{Remark}

The bilinear equations \eqref{bilin-1III_3}, \eqref{bilin-4III_3} are also proved by using the Nakajima--Yoshioka blow-up equations~\cite{BS1}. There exists another $q$-difference equation admitting $\mathrm{P_{III_3}}$ and $\mathrm{P_I}$ as limits~\cite{GR}, which corresponds to the $q$-difference Painlev\'e equation of the surface type $A_7^{(1)}$ \cite{S}. Its standard form (see equation~(2.44) in \cite{S1}) is
\begin{gather}\label{eq_qPA7}
\overline{g}g^2\underline{g}=t^{2}(1-g),
\end{gather}
where $g=g(t)$. A series expansion of the tau function for $q$-$P\big(A_7^{(1)}\big)$ \eqref{eq_qPA7} was proposed and conjectured to satisfy its bilinear form in \cite{BGM}. Later, it was proved in~\cite{BS1}. Below, we show that their tau function for $q$-$P\big(A_7^{(1)}\big)$~\eqref{eq_qPA7} is also obtained as another limit of the tau function for $q$-$\mathrm{P_{III_2}}$.

Redefine the tau function by
\begin{gather*}
\tau^{\mathrm{III_3}}(s,\sigma,t)= \sum_{n\in\mathbb{Z}} s^n t^{(\sigma+n)^2} C_{\mathrm{{III}_3}} [\sigma+n ] Z_{\mathrm{{III}_3}}[\sigma+n,t],
\end{gather*}
with
\begin{gather*}
C_{\mathrm{{III}_3}} [\sigma ] = (-1)^{n^2}(q-1)^{-4\sigma^2}\prod_{\varepsilon=\pm} \frac{1}{G_q(1+2\varepsilon(\sigma+n))},\\
Z_{\mathrm{{III_3}}}[\sigma,t]=\sum_{(\lambda_+,\lambda_-)\in\mathbb{Y}^2}t^{|\lambda_+|+|\lambda_-|}
\frac{\prod\limits_{\varepsilon=\pm}f_{\lambda_\varepsilon}(q^{\varepsilon\sigma})^{-1}}
{\prod\limits_{\varepsilon,\varepsilon'=\pm}N_{\lambda_\varepsilon,\lambda_{\varepsilon'}}\big(q^{(\varepsilon-\varepsilon')\sigma}\big)}.
\end{gather*}

Let us define the tau functions for $q$-$P\big(A_7^{(1)}\big)$ by
\begin{gather*}
\tau^{\mathrm{III_3}}_1=\tau^{\mathrm{III_3}}\big(s,\sigma,t/\sqrt{q}\big),\qquad
\tau^{\mathrm{III_3}}_2=\tau^{\mathrm{III_3}}\big(s,\sigma+\tfrac{1}{2},t\big).
\end{gather*}
Put
\begin{gather*}
C_1= (q-1)^{-\sigma^2}\prod_{\varepsilon=\pm}G_q\big( \tfrac{3}{2}+\Lambda+\varepsilon\sigma\big)^{-1}, \\
C_2=(q-1)^{-(\sigma+1)^2}\prod_{\varepsilon=\pm}G_q\big(\tfrac{1}{2}+\Lambda+\varepsilon(\sigma+1)\big)^{-1},\\
C_3=(q-1)^{-(\sigma+1/2)^2}\prod_{\varepsilon=\pm}G_q\big(1+\Lambda+\varepsilon\big(\sigma+\tfrac{1}{2}\big)\big)^{-1}.
\end{gather*}

\begin{Proposition}Set
\begin{gather*}
\theta_*=-\Lambda, \qquad s = \tilde s (q-1)^{-2\sigma}\prod_{\varepsilon=\pm}\Gamma_q\big(\tfrac{1}{2}+\Lambda+\varepsilon\sigma\big)^{-\varepsilon}.
\end{gather*}
Then we have
\begin{gather*}
C_1\tau_1^{\mathrm{III_2}}(\theta_*\,|\, s,\sigma,t) \to \tau_1^{\mathrm{III_3}}(\tilde s, \sigma, t),\\
C_2\tau_2^{\mathrm{III_2}}(\theta_*\,|\, s,\sigma,t)\to \tau_1^{\mathrm{III_3}}(\tilde s, \sigma, q t)/\tilde s,\\
C_3\tau_3^{\mathrm{III_2}}(\theta_*\,|\, s,\sigma,t)\to \tau_2^{\mathrm{III_3}}(\tilde s, \sigma, t),
\end{gather*}
as $\Lambda\to\infty$.
\end{Proposition}

In what follows, we abbreviate $\tau_i^{\mathrm{III_3}}( s,\sigma,t)$ to $\tau_i$. Since we have the relation
\begin{gather*}
 C_1C_2=\frac{(q-1)^{1/2}}{1-q^{\Lambda-\sigma+1/2}}C_3^2,
\end{gather*}
we obtain the following theorem by the degeneration.
\begin{Theorem}
The functions
\begin{gather*}
y=-q^{-1}t^{1/2}\frac{s\tau_2^2}{\tau_1\overline{\tau_1}},\qquad
z=-q^{-3/4}t^{1/2}\frac{s\tau_2\underline{\tau_2}}{\underline{\tau_1}\overline{\tau_1}}
\end{gather*}
solves
\begin{gather}
y \overline y=-q^{-3/2}\frac{\overline z^2}{\overline z-q^{-1/2}} ,\qquad z \overline z=y(qy- t).\label{eq_qPIII_3_2}
\end{gather}
 Furthermore, the tau functions $\tau_1$, $\tau_2$ satisfy the following bilinear equations.
\begin{gather}
 s^{-1}t^{1/2} \tau_1\overline{\tau_1}-\tau_2^2 +\overline{\tau_2}\underline{\tau_2}=0,\label{bilin-4III_3_2}\\
 \tau_1^2-sq^{-1/4}t^{1/2}\tau_2\underline{\tau_2}-\underline{\tau_1}\overline{\tau_1}=0.\label{bilin-5III_3_2}
\end{gather}
\end{Theorem}
We note that the bilinear equations \eqref{bilin-4III_2}, \eqref{bilin-5III_2} for the tau functions of $q$-$\mathrm{P_{III_2}}$ degenerate to \eqref{bilin-4III_3_2}, \eqref{bilin-5III_3_2}, respectively. By the change of variables $t\to \sqrt{q}t$, $\sigma\to\sigma+1/2$, the bilinear equation~\eqref{bilin-5III_3_2} transforms~\eqref{bilin-4III_3_2}. The bilinear equation~\eqref{bilin-4III_3_2} is equivalent to the bilinear equation~(4.20) for $N=2$, $m=1$ in~\cite{BS1}, which is for
$q$-$P\big(A_7^{(1)}\big)$. Following \cite[Example~3.5]{BGM}, we take a~time evolution $T$ as $T(f(\sigma,t))=f(\sigma+1/2,\sqrt{q} t)$. Then the bilinear equation \eqref{bilin-5III_3_2} is equivalent to
\begin{gather*}
\tau^2-t^{1/2}\overline{\tau}\underline{\tau}-\overline{\overline{\tau}} \underline{\underline{\tau}}=0,
\end{gather*}
where $\tau=\tau^{\mathrm{III_3}}(s,\sigma,t)$, $\overline{\tau}=T(\tau)$, $\underline{\tau}=T^{-1}(\tau)$. Let $g=t^{1/2}\overline{\tau}\underline{\tau}\tau^{-2}$, then $g$ satisfies $q$-$P\big(A_7^{(1)}\big)$~\eqref{eq_qPA7}.

\subsection*{Acknowledgements}

This work is partially supported by JSPS KAKENHI Grant Number JP15K17560. The authors thank the referees for valuable suggestions and comments.

\pdfbookmark[1]{References}{ref}
\LastPageEnding

\end{document}